\documentclass[conference]{IEEEtran}
\usepackage{hyperref} 
\usepackage{amsmath,amssymb,amsthm,cite,graphicx,array}
\theoremstyle{definition}
\newtheorem{theorem}{Theorem}
\newtheorem{lemma}{Lemma}
\newtheorem{definition}{Definition}

\newtheorem{remark}{Remark}
\newtheorem{example}{Example}
\usepackage{color}

\title{Optimal Scalar Linear Index Codes for One-Sided Neighboring Side-Information Problems}
\begin{document}
\author{Mahesh~Babu~Vaddi~and~B.~Sundar~Rajan\\ 
 Department of Electrical Communication Engineering, Indian Institute of Science, Bengaluru 560012, India \\ E-mail:~\{mahesh,~bsrajan\}@ece.iisc.ernet.in }
\maketitle
\begin{abstract}
The capacity of symmetric instance of the multiple unicast index coding problem with neighboring antidotes (side-information) with number of messages equal to the number of receivers was given by Maleki \textit{et al.}  In this paper, we construct matrices of size  $m \times n~(m \geq n)$ over $\mathbb{F}_q$ such that any $n$ adjacent rows of the matrix are linearly independent. By using such matrices, we give an optimal scalar linear index codes over $\mathbb{F}_q$ for the symmetric one-sided antidote problems considered by Maleki \textit{et al.} for any given number of messages and one-sided antidotes. The constructed codes are independent of field size and hence works over every field.
\end{abstract}
\section{Introduction and Background}
\label{sec1}
An index coding problem consists of one transmitter, $M$ receivers $\{R_1,R_2,\ldots,R_M\}$ and $K$ independent messages 
$\{x_1,x_2,\ldots,x_K\}$, where $x_i\in\mathbb{F}_q^{p_i},~ x_i=(x_{i,1},x_{i,2},\ldots,x_{i, p_i }),~x_{i,j} \in \mathbb{F}_q$ for 
$i \in \{1,2,\ldots,K \}$ and $j \in \{1,2,\ldots,p_i\}$. Each receiver $R_i$ is identified with $\{\mathcal{W}_i,\mathcal{K}_i\}$, where $\mathcal{W}_i \subseteq \{x_1,x_2,\ldots,x_K\}$ is the set of wanted messages and $\mathcal{K}_i \subseteq \{x_1,x_2,\ldots,x_K\}$ is the set of known messages to receiver $R_i$. The messages in the set $\mathcal{K}_i$ are also called side information or antidotes to receiver $R_i$. The transmitter has all the $K$ messages and it also knows the set of wanted and known messages of each receiver. An index code is a mapping defined as follows:
\begin{align*}
\mathfrak{C}: \mathbb{F}^{p_1+p_2+\ldots+p_K}_q \rightarrow \mathbb{F}^N_q,
\end{align*}
where $N$ is the length of index code. That is, the index code $\mathfrak{C}$ maps $K$ messages $x_1,x_2,\ldots,x_K$ into $N$ code symbols $y_1$,$y_2,\ldots$,$y_N$ ($y_i\in \mathbb{F}_q$ for $i=1,2,\ldots,N$). If $p_1=p_2=\cdots=p_K$, then the index code is called symmetric rate vector index code. If $p_1=p_2=\cdots=p_K=1$, then the index code is called scalar index code. The index coding problem is to design an index code such that the number of transmissions broadcasted by the transmitter is minimized and all the receivers get their wanted messages by using the index code broadcasted and their known information.

Instead of one transmitter and $M$ receivers, the index coding problem can also be viewed as $M$ source-receiver pairs with all $M$ sources connected with all $M$ receivers through a common finite capacity channel \cite{MCJ} and all source-receiver pairs connected with either zero of infinite capacity channels. This problem is called multiple unicast index coding problem (MUICP).

The problem of index coding with side information was introduced by Birk and Kol \cite{BiK}. Bar-Yossef \textit{et al.} \cite{YBJK} studied the class of index coding problems in which each receiver demands only a unique message and the number of receivers equals the number of messages. Ong and Ho \cite{OnH} classified binary index coding problem depending on the demands and side information of the receivers. An index coding problem is called unicast if the demand sets of the receivers are disjoint. 
For the unicast index coding problem, it was shown that the length of an optimal linear index code is equal to the minrank of the side information graph \cite{YBJK} of the index coding problem and finding the minrank is NP hard \cite{minrank}. 

Maleki \textit{et al.} \cite{MCJ} found the capacity of symmetric MUICP with neighboring antidotes. In a symmetric MUICP with equal number of $K$ messages and source-receiver pairs, each receiver has a total of $U+D=A<K$ antidotes, corresponding to the $U$ messages before and $D$ messages after its desired message. In this setting, the $k$th receiver $R_{k}$ demands the message $x_{k}$ having the antidotes
\begin{equation}
\label{antidote}
{\cal K}_k= \{x_{k-U},\dots,x_{k-2},x_{k-1}\}\cup\{x_{k+1}, x_{k+2},\dots,x_{k+D}\}.
\end{equation}

The symmetric capacity of this index coding problem setting is
\begin{flushleft}
$C=\left\{
                \begin{array}{ll}  
                  {~~~1\qquad\quad\ ~~~~  \mbox{if}~~~~A=K-1}\\
                  {\frac{U+1}{K-A+2U}} ~~~~~~~ \mbox{if}~~~~A\leq K-2\qquad $per message,$
                  \end{array}
              \right.$
              \end{flushleft}
               where $U,D \in$ $\mathbb{Z},$ $0 \leq U \leq D$, and $U+D=A<K$.
\vspace{-10pt}
\begin{align}
\label{capacity}
\end{align}

In the setting given in \cite{MCJ} with one-sided antidote cases, i.e., the cases where $U$  is zero, 
the $k$th receiver $R_{k}$ demands the message $x_{k}$ having the antidotes,
\begin{equation}
\label{antidote1}
{\cal K}_k =\{x_{k+1}, x_{k+2},\dots,x_{k+D}\}, 
\end{equation}
\noindent
for which \eqref{capacity} reduces to
\begin{equation}
\label{capacity1}
C=\left\{
                \begin{array}{ll}
                  {~~1 ~~~~~~~~~~ \mbox{if} ~~ D=K-1}\\
                  {\frac{1}{K-D}} ~~~~~~~ \mbox{if} ~~D\leq K-2  ~~~~~\text{per message}. 
                  \end{array}
              \right.
\end{equation}

The set of messages neither known nor demanded by $R_k$ contribute interference ($\mathcal{I}_k$) at the receiver $R_k$, where
\begin{align}
\label{interference}
\mathcal{I}_k=\{x_{1},x_{2},\ldots,x_{k-1},x_{k+D+1},x_{k+D+2},\ldots,x_K\}.
\end{align}
 
The reciprocal of the capacity is called the optimal length of an index coding problem. That is, at least $\frac{1}{C}$ code symbols are required to convey one wanted message to each receiver. 

In a scalar linear code, the messages in an index coding problem take value from a finite field $\mathbb{F}_q$. A $K$-tuple $(x_1,x_2,\ldots,x_K)\in\mathbb{F}_q^K$ of messages is denoted by $\mathbf{x}$. A scalar linear index code of length $N$ $(<K)$ is represented by an encoding matrix $\mathbf{L}$ $(\in \mathbb{F}_q^{K\times N})$, where the $j$th column contains the coefficients used for mixing messages $x_1,x_2,\ldots,x_K$ to get the $j$th code symbol and the $i$th row $L_i$ $(\in \mathbb{F}_q^{1\times N})$ contains the coefficients used for mixing message $x_i$ in the $N$ code symbols. A codeword of the index code is 
\begin{align*}
 [y_1~y_2~\ldots~y_N]=\mathbf{xL}=\sum_{i=1}^{K}x_iL_i.
\end{align*}

In this paper, the $m\times m$ identity matrix is denoted by $\mathbf{I}_{m}$. For a subset $I=\{i_1,i_2,\ldots,i_l\} \subseteq \{1,2,\ldots,K\}$, let $x_I=\{x_{i_1},x_{i_2},\ldots,x_{i_l}\}$ and $L_I=\{L_{i_1},L_{i_2},\ldots,L_{i_l}\}$. All the subscripts in this paper are to be considered $modulo \ K$. The set of rows $\{L_{k+1},L_{k+2},\cdots,L_{k+N}\}$  for $k=1,2,\ldots,K$ (all subscripts are $modulo \ K$) are called adjacent rows in the matrix $L_{K \times N}$.
\subsection{Contributions}
\begin{itemize}
\item In this paper, we give a construction of $0-1$ (binary) matrices with a given size $m \times n~(m \geq n)$, such that any $n$ adjacent rows in the matrix are linearly independent over every field $\mathbb{F}_q$. 
\medskip
\item For the neighboring antidote symmetric MUICP, Maleki \textit{et al.} \cite{MCJ} proved the existence of capacity achieving codes by using Vandermonde matrices over large fields. The size of the field in their construction depends on the number of messages $K$. In this paper, using the proposed matrix construction, we give capacity achieving scalar linear codes for given $K$ and $D$ over every field $\mathbb{F}_q$, which is independent of $K$. 
\medskip
\item In \cite{MRarXiv}, we proposed a vector linear index code construction that constructs a sequence of MUICPs with two-sided antidotes with a vector linear index code starting from a given one-sided antidote MUICP with a known scalar linear index code. The construction given in this paper along with the construction in \cite{MRarXiv} gives a capacity achieving vector linear code for two-sided neighboring antidote problems with every $K,U$ and $D$ and the constructed codes are independent of field size.
\medskip

\end{itemize}

In \cite{MRRarXiv}, we proposed the construction of capacity achieving scalar linear index codes for one-sided antidote problem for $K$ and $D$ satisfying some conditions. In this paper, we construct capacity achieving scalar linear index codes for arbitrary $K$ and $D$ and the constructed codes are independent of field size.


\section{Construction of the optimal length index codes over $\mathbb{F}_q$}
In this section, we give a method to construct a $K \times (K-D)$ encoding matrix for one-sided neighboring antidote MUICP with $K$ messages and $D$ antidotes. 
\begin{lemma}
\label{lemma1}
Consider a MUICP with $K$ messages and  $K$ receivers. Receiver $R_k$ wants the message $x_k$ and its antidotes and interference are given by \eqref{antidote1} and  \eqref{interference} respectively. Let $\mathbf{L}$ be a $K\times (K-D)$ encoding matrix for this index coding problem. Then, the receiver $R_k$ can decode $x_k$ if and only if 
\begin{align} 
\label{ind}
L_{k}\notin \mathsf{span}\; L_{{\cal{I}}_k} ~~ \text{for~each} ~~k\in\{1,2,\ldots,K\},
\end{align}
 where $L_{{\cal{I}}_k}=\{L_{1},L_2,\ldots,L_{k-1},L_{k+D+1},L_{k+D+2},\ldots,L_{K}\}.$
\end{lemma}
\begin{proof}
The received vector $y$ can be written as
\begin{align}
\nonumber
y&=x_{{\cal{K}}_k}L_{{\cal{K}}_k}+x_kL_k+x_{{\cal{I}}_k}L_{{\cal{I}}_k}, \text{ or}
\\ z&=y-x_{{\cal{K}}_k}L_{{\cal{K}}_k}=x_kL_k+x_{\mathcal{I}_k}L_{\mathcal{I}_k},
\label{decoding}
\end{align}
where $z$ can be computed by $R_k$ using its antidotes $x_{{{\cal{K}}_k}}$. 

Assume that \eqref{ind} is satisfied for $k$. This implies that $L_k$ is not in the span of $\{L_{1},L_2,\ldots,L_{k-1},L_{k+D+1},\ldots,L_{K}\}$. Then, $z$ can be expressed as the following linear combination
\begin{align*}
z=&a_kL_{k}+a_{1}L_{1}+a_{2}L_{2}+\ldots+a_{k-1}L_{k-1} 
\\&+a_{k+D+1}L_{k+D+1}+a_{k+D+2}L_{k+D+2}+\ldots+a_{K}L_{K},
\end{align*}
where $a_k$ is unique and $x_k=a_k$. 

If, on the contrary, $L_{k}$ does not satisfy \eqref{ind}, then $a_{k}$ will no longer be unique and consequently $x_{k}$ can not be decoded by $R_k$. This implies that $\mathbf{L}$ is not an index code encoding matrix which contradicts the assumption. This completes the proof.
\end{proof}
\begin{lemma}
\label{lemma2}
In the index coding problem mentioned in Lemma \ref{lemma1}, if every $K-D$ adjacent rows of the encoding matrix $\mathbf{L}$ are linearly independent, then the receiver $R_k$ can decode $x_k$ and all $K-D-1$ interfering messages in ${{\cal{I}}_k}$ for $k=1,2,\ldots,K$.
\end{lemma}
\begin{proof}
If every $K-D$ adjacent rows of the encoding matrix $\mathbf{L}$ are linearly independent, then \eqref{decoding} is a set of $K-D$ linearly independent equations with $K-D$ unknowns. The unknowns in \eqref{decoding} are $x_k$ and $K-D-1$ interfering messages in ${{\cal{I}}_k}$. Hence, by solving $K-D$ equations in \eqref{decoding}, the receiver $R_k$ can decode $x_k$ and all $K-D-1$ interfering messages in ${{\cal{I}}_k}$ for $k=1,2,\ldots,K$.
\end{proof}
Rectangular circulant matrix defined below is used in the construction of matrices with a given size $m \times n~(m \geq n)$, such that any $n$ adjacent rows in the matrix are linearly independent over every field $\mathbb{F}_q$.
\begin{definition}
\label{def1}
Let $\lambda$ and $\mu$ be two positive integers and $\lambda$ divides $\mu$. The following rectangular circulant matrix is denoted by $\mathbf{C}_{\mu \times \lambda}$. 
\end{definition}
$$\mathbf{C}_{\mu \times \lambda}=\left.\left[\begin{array}{*{20}c}
   \mathbf{I}_\lambda  \\
   \mathbf{I}_\lambda  \\
   \vdots  \\
   \mathbf{I}_\lambda   \\
   \mathbf{I}_\lambda 
   \end{array}\right]\right\rbrace \frac{\mu}{\lambda}~\text{number~of}~ \mathbf{I}_\lambda~\text{matrices}$$

If $\lambda=\mu$, then $\mathbf{C}_{\mu \times \lambda}=\mathbf{I}_{\lambda}$. In the matrix $\mathbf{C}_{\mu \times \lambda}$, every set of $\lambda$ adjacent rows are linearly independent over every field $\mathbb{F}_q$. In the rectangular circulant matrix $\mathbf{C}_{\mu \times \lambda}^{\mathsf{T}}$, every set of $\lambda$ adjacent columns are linearly independent over every field $\mathbb{F}_q$. Let the matrix $\mathbf{D}_{\mu \times (t\mu+\lambda)}$ (for some integer $t$) be the column concatenation of $\mathbf{C}_{\mu \times \lambda}$ with $t$ identity matrices  $\mathbf{I}_{\mu}$. In the matrix $\mathbf{D}_{\mu \times (t\mu+\lambda)}$, every set of $\mu$ adjacent columns are linearly independent over every field $\mathbb{F}_q$.
%

\subsection{Construction of an encoding matrix $\mathbf{L}_{K \times (K-D)}$}
In this subsection we present our construction of encoding matrices which will be referred as {\bf CONSTRUCTION} henceforth. 

\begin{center}
{\bf CONSTRUCTION}
\end{center}

For a given $K$ and $D$ let  
\begin{align}
\nonumber
&\lambda_1=(K-D)~modulo~D,
\nonumber
\\& \lambda_2=D~modulo~\lambda_1,
\nonumber
\\& \lambda_3= \lambda_1~modulo~\lambda_2,
\nonumber
\\& ~~~~\cdots
\nonumber
\\& \lambda_i= \lambda_{i-2}~modulo~\lambda_{i-1}
\nonumber
\\& ~~~~\cdots
\nonumber
\\& \lambda_{l}= \lambda_{l-2}~modulo~\lambda_{l-1}
\label{chain}
\end{align}

\noindent
where $\lambda_l$ divides $ \lambda_{l-1}$ for some integer $l.$

Depending upon whether $l$ is even or odd we have the following two cases.

{\bf Case I: $l$ is an even integer}
For this case the structure of the encoding matrix $\mathbf{L}_{K \times (K-D)}$ in terms of rectangular circulant matrices is shown in Fig. \ref{fig1} and in Fig. \ref{fig3}. Both these figures give the same encoding matrix.

\begin{itemize}
\item Start with the rectangular circulant matrix $\mathbf{C}_{\lambda_{l-1} \times \lambda_{l}}^{\mathsf{T}}\triangleq\mathbf{D}_{\lambda_l \times \lambda_{l-1}}^{(1)}$ (this is a rectangular circulant matrix because $\lambda_l$ divides $\lambda_{l-1}$).
\item Concatenate the rectangular circulant matrix $\mathbf{C}_{(\lambda_{l-2}-\lambda_l)\times \lambda_{l-1}}$ ($\lambda_{l-1}$ divides $(\lambda_{l-2}-\lambda_l)$) to the rows of the matrix $\mathbf{D}_{\lambda_l \times \lambda_{l-1}}^{(1)}$ to obtain the concatenated matrix of size $\lambda_{l-2} \times \lambda_{l-1}$. Let this concatenated matrix be $\mathbf{D}_{\lambda_{l-2} \times \lambda_{l-1}}^{(2)}.$
\item Concatenate the rectangular circulant matrix $\mathbf{C}_{(\lambda_{l-3}-\lambda_{l-1})\times \lambda_{l-2} }^{\mathsf{T}}$ ($\lambda_{l-2}$ divides $(\lambda_{l-3}-\lambda_{l-1})$) to the columns of the matrix $\mathbf{D}_{\lambda_{l-2} \times \lambda_{l-1}}^{(2)}$ to obtain the matrix of size $\lambda_{l-2} \times \lambda_{l-3}$. Let this concatenated matrix be $\mathbf{D}_{\lambda_{l-2} \times \lambda_{l-3}}^{(3)}.$
\item Repeat the above procedure until we get $K \times (K-D)$ matrix. The construction of $K \times (K-D)$ matrix by the above procedure is guaranteed by \eqref{chain}. 
\item The sequence of construction of the matrices can be summarized as given below:
\begin{align}
\label{seq1}
\nonumber
&\mathbf{C}_{\lambda_{l-1} \times \lambda_{l}}^{\mathsf{T}}\triangleq\mathbf{D}_{\lambda_l \times \lambda_{l-1}}^{(1)}\rightarrow \mathbf{D}_{\lambda_{l-2} \times \lambda_{l-1}}^{(2)}\rightarrow \cdots \rightarrow\\&\mathbf{D}_{\lambda_{l-t} \times \lambda_{l-t+1}}^{(t)}
\rightarrow \cdots \rightarrow \mathbf{D}_{\lambda_{2} \times \lambda_{3}}^{(l-2)}\rightarrow \mathbf{D}_{\lambda_{2} \times \lambda_{1}}^{(l-1)}\rightarrow
\mathbf{D}_{D \times \lambda_{1}}^{(l)}
\nonumber
\\&\rightarrow \mathbf{D}_{D \times (K-D)}^{(l+1)}\rightarrow \mathbf{D}_{K \times (K-D)}^{(l+2)}=\mathbf{L}_{K \times (K-D)} .
\end{align}
\end{itemize}

If $\lambda_l$ divides $\lambda_{l-1}$ and $l$ is an even number, then the construction starts with a fat (number of columns greater than number of rows) matrix and the construction proceeds by constructing alternate tall (number of rows greater than number of columns) and fat matrices.

{\bf Case II: $l$ is an odd integer}

For this case the structure of the encoding matrix $\mathbf{L}_{K \times (K-D)}$ in terms of rectangular circulant matrices is shown in Fig. \ref{fig2} and in Fig. \ref{fig3}. Both these figures give the same encoding matrix.

\begin{itemize}
\item Start with the rectangular circulant matrix $\mathbf{C}_{\lambda_{l-1} \times \lambda_{l}}\triangleq\mathbf{D}_{\lambda_{l-1} \times \lambda_{l}}^{(1)}$ (this is a rectangular circulant matrix because $\lambda_l$ divides $\lambda_{l-1}$).
\item Concatenate the matrix $\mathbf{C}_{(\lambda_{l-2}-\lambda_l) \times \lambda_{l-1}}^{\mathsf{T}}$ to the columns of the matrix $\mathbf{D}_{\lambda_{l-1} \times \lambda_{l}}^{(1)}$ to obtain the concatenated matrix of size $\lambda_{l-1} \times \lambda_{l-2}$. Let this concatenated matrix be $\mathbf{D}_{\lambda_{l-1} \times \lambda_{l-2}}^{(2)}.$
\item Concatenate the $\mathbf{C}_{(\lambda_{l-3}-\lambda_{l-1}) \times \lambda_{l-2}}$ to the rows of the matrix $\mathbf{D}_{\lambda_{l-1} \times \lambda_{l-2}}^{(2)}$ to obtain the matrix of size $\lambda_{l-3} \times \lambda_{l-2}$. Let this concatenated matrix be $\mathbf{D}_{\lambda_{l-3} \times \lambda_{l-2}}^{(3)}.$
\item Repeat the above procedure until we get $K \times (K-D)$ matrix. The construction of $K \times (K-D)$ matrix by the above procedure is guaranteed by \eqref{chain}. 
\item The sequence of construction of the matrices can be summarized as given below:
\begin{align}
\label{seq2}
\nonumber
&\mathbf{C}_{\lambda_{l-1} \times \lambda_{l}}\triangleq\mathbf{D}_{\lambda_{l-1} \times \lambda_{l}}^{(1)}\rightarrow \mathbf{D}_{\lambda_{l-1} \times \lambda_{l-2}}^{(2)}\rightarrow \cdots \rightarrow \\&
\mathbf{D}_{\lambda_{l-t} \times \lambda_{l-t+1}}^{(t)}
\rightarrow \cdots \rightarrow \mathbf{D}_{\lambda_{2} \times \lambda_{3}}^{(l-2)}\rightarrow \mathbf{D}_{\lambda_{2} \times \lambda_{1}}^{(l-1)}\rightarrow
\mathbf{D}_{D \times \lambda_{1}}^{(l)}
\nonumber
\\&\rightarrow \mathbf{D}_{D \times (K-D)}^{(l+1)}\rightarrow \mathbf{D}_{K \times (K-D)}^{(l+2)}\rightarrow\mathbf{L}_{K \times (K-D)}.
\end{align}
\end{itemize}

If $\lambda_l$ divides $\lambda_{l-1}$ and $l$ is an odd number, then the construction starts with a tall matrix and proceeds by constructing alternate fat and tall matrices.
\begin{figure*}
\centering
\includegraphics[scale=0.60]{matrix4.eps}\\
\caption{Encoding matrix if $l$ is even.}
\label{fig1}
\end{figure*}

\begin{figure*}
\centering
\includegraphics[scale=0.60]{matrix3.eps}\\
\caption{Encoding matrix if $l$ is odd.}
\label{fig2}
\end{figure*}
\begin{figure*}
\centering
\includegraphics[scale=0.56]{matrix2.eps}\\
~If $\lambda_l$ divides $\lambda_{l-1}$ and $l$ is an even integer, then $\mathbf{S}=\mathbf{C}_{\lambda_{l-1} \times \lambda_{l}}^{\mathsf{T}}$.\\
If $\lambda_l$ divides $\lambda_{l-1}$ and $l$ is an odd integer, then ~$\mathbf{S}=\mathbf{C}_{\lambda_{l-1} \times \lambda_{l}}$
\caption{Encoding matrix.}
\label{fig3}
\end{figure*}
\begin{lemma}
\label{lemma3}
Let $\mathbf{D}_{\lambda_{l-t} \times \lambda_{l-t+1}}^{(t)}$  be the binary matrix obtained during CONSTRUCTION for a given $K$,$D$ and for some integer $1 \leq t \leq l+2$. Note that the matrix $\mathbf{D}_{\lambda_{l-t} \times \lambda_{l-t-1}}^{(t+1)}$ which is $[\mathbf{I}_{\lambda_{l-t}}:\mathbf{I}_{\lambda_{l-t}}:\cdots:\mathbf{I}_{\lambda_{l-t}}:\mathbf{D}_{\lambda_{l-t} \times \lambda_{l-t+1}}^{(t)}]$  is a column concatenation of $t$ identity matrices of size $\lambda_{l-t} \times \lambda_{l-t}$ with the matrix $\mathbf{D}_{\lambda_{l-t} \times \lambda_{l-t+1}}^{(t)}$. In the matrix $\mathbf{D}_{\lambda_{l-t} \times \lambda_{l-t-1}}^{(t+1)}$, ~every $\lambda_{l-t}$ adjacent columns are linearly independent over every field $\mathbb{F}_q$.
\end{lemma}
\begin{proof}
The proof is by induction on $t$. For $t=1$, the columns of the matrix $\mathbf{D}_{\lambda_{l-1} \times \lambda_{l}}^{(1)}$ are independent which follows from the CONSTRUCTION. 

Let the Lemma be true for $t$. By induction hypothesis every set of $\lambda_{l-t+1}$ rows are independent in $\mathbf{D}_{\lambda_{l-t} \times \lambda_{l-t+1}}^{(t)}$ and we prove that every set of $\lambda_{l-t}$ columns are independent in $\mathbf{D}_{\lambda_{l-t} \times \lambda_{l-t-1}}^{(t+1)}$ as follows. 

The matrix $\mathbf{D}_{\lambda_{l-t} \times \lambda_{l-t-1}}^{(t+1)}$ comprises of $\lambda_{l-t-1}$ columns. If we select $\lambda_{l-t}$ adjacent columns from the first $\lambda_{l-t-1}-\lambda_{l-t}$ columns of the matrix $\mathbf{D}_{\lambda_{l-t} \times \lambda_{l-t-1}}^{(t+1)}$, the selected columns are columns of the $\mathbf{I}_{\lambda_{l-t}}$ (in a different order) and hence they are linearly independent. 

If we select $\lambda_{l-t}$ adjacent columns in such a way that $\lambda_{l-t}-m$ columns from $\mathbf{I}_{\lambda_{l-t}}$ and $m$ columns from $\mathbf{D}_{\lambda_{l-t} \times \lambda_{l-t+1}}^{(t)}$ ($m \leq \lambda_{l-t+1}$), then the selected $\lambda_{l-t}$ columns can be written as $\lambda_{l-t} \times \lambda_{l-t}$ matrix as given below
\begin{align} 
\label{matrix3}
\begin{bmatrix}
\textbf{0}&\mathbf{B}\\
\mathbf{A}&\mathbf{C}\\
\end{bmatrix},
\end{align}
where $[\mathbf{0}~\mathbf{A}]^\mathsf{T}$ is a $\lambda_{l-t} \times (\lambda_{l-t}-m)$ matrix corresponding to $\lambda_{l-t}-m$ columns of  $\mathbf{I}_{\lambda_{l-t}}$ and $[\mathbf{B}~\mathbf{C}]^\mathsf{T}$ is a $\lambda_{l-t} \times m$ matrix corresponding to $m$ columns of $\mathbf{D}_{\lambda_{l-t} \times \lambda_{l-t+1}}^{(t)}$. The matrix $\mathbf{A}$ is an identity matrix of size $(\lambda_{l-t}-m) \times (\lambda_{l-t}-m)$. 
The matrix $\mathbf{B}$ of dimension $m \times m$ is an identity matrix because the CONSTRUCTION gaurantees that the first $\lambda_{l-t+1}$ rows in $\mathbf{D}_{\lambda_{l-t} \times \lambda_{l-t+1}}^{(t)}$  are $\mathbf{I}_{\lambda_{l-t+1}}$. The matrix given in \eqref{matrix3} is a full rank matrix follows from the fact that it is a triangular matrix with $\mathbf{A}$ and $\mathbf{B}$ being identity matrices over every field $\mathbb{F}_q$.

If we select $\lambda_{l-t}$ adjacent columns in such a way that $m$ columns from $\mathbf{D}_{\lambda_{l-t} \times \lambda_{l-t+1}}^{(t)}$ ($m \leq \lambda_{l-t}$) and $\lambda_{l-t}-m$ columns from $\mathbf{I}_{\lambda_{l-t}}$, then the selected $\lambda_{l-t}$ columns can be written as $\lambda_{l-t} \times \lambda_{l-t}$ matrix as given below
\begin{align} 
\label{matrix4}
\begin{bmatrix}
\textbf{E}&\mathbf{G}\\
\mathbf{F}&\mathbf{0}\\
\end{bmatrix},
\end{align}
where $[\mathbf{G}~\mathbf{0}]^\mathsf{T}$ is a $\lambda_{l-t} \times (\lambda_{l-t}-m)$ matrix corresponding to $\lambda_{l-t}-m$ columns of  $\mathbf{I}_{\lambda_{l-t}}$ and $[\mathbf{E}~\mathbf{F}]^\mathsf{T}$ is a $\lambda_{l-t} \times m$ matrix corresponding to $m$ columns of $\mathbf{D}_{\lambda_{l-t} \times \lambda_{l-t+1}}^{(t)}$. The matrix $\mathbf{G}$ is an identity matrix of size $(\lambda_{l-t}-m) \times (\lambda_{l-t}-m)$. The matrix $\mathbf{F}$ is a full rank matrix follows from induction hypothesis.
The matrix given in \eqref{matrix4} is a full rank matrix follows from the fact that it is a triangular matrix with $\mathbf{G}$ and $\mathbf{F}$ having full rank over every field $\mathbb{F}_q$.
\end{proof}
\subsection{Examples}
\begin{example}
Consider a one-sided neighboring antidote MUICP with $K=21, D=4$. For this index coding problem, we have $\lambda_1=1$ and $l=1.$ 

An encoding matrix $\mathbf{L}_{21 \times 17}$ for this index coding problem is obtained in the following steps $\mathbf{C}_{4 \times 1}\triangleq\mathbf{D}_{4 \times 1}^{(1)} 
 \rightarrow \mathbf{D}_{4 \times 17}^{(2)} \rightarrow \mathbf{D}_{21 \times 17}=\mathbf{L}_{21 \times 17}$ by concatenating suitable rectangular circulant matrices.
The encoding matrix $\mathbf{L}_{21 \times 17}$ is given below.

\begin{small}
\arraycolsep=3pt
\setlength\extrarowheight{-3.0pt}
{
$$L_{21 \times 17}=\left[\begin{array}{*{20}c}
   1 & 0 & 0 & 0 & 0 & 0 & 0 & 0 & 0 & 0 & 0 & 0 & 0 & 0 & 0 & 0 & 0 \\
   0 & 1 & 0 & 0 & 0 & 0 & 0 & 0 & 0 & 0 & 0 & 0 & 0 & 0 & 0 & 0 & 0 \\
   0 & 0 & 1 & 0 & 0 & 0 & 0 & 0 & 0 & 0 & 0 & 0 & 0 & 0 & 0 & 0 & 0 \\
   0 & 0 & 0 & 1 & 0 & 0 & 0 & 0 & 0 & 0 & 0 & 0 & 0 & 0 & 0 & 0 & 0 \\
   0 & 0 & 0 & 0 & 1 & 0 & 0 & 0 & 0 & 0 & 0 & 0 & 0 & 0 & 0 & 0 & 0 \\
   0 & 0 & 0 & 0 & 0 & 1 & 0 & 0 & 0 & 0 & 0 & 0 & 0 & 0 & 0 & 0 & 0 \\
   0 & 0 & 0 & 0 & 0 & 0 & 1 & 0 & 0 & 0 & 0 & 0 & 0 & 0 & 0 & 0 & 0 \\
   0 & 0 & 0 & 0 & 0 & 0 & 0 & 1 & 0 & 0 & 0 & 0 & 0 & 0 & 0 & 0 & 0 \\
   0 & 0 & 0 & 0 & 0 & 0 & 0 & 0 & 1 & 0 & 0 & 0 & 0 & 0 & 0 & 0 & 0 \\
   0 & 0 & 0 & 0 & 0 & 0 & 0 & 0 & 0 & 1 & 0 & 0 & 0 & 0 & 0 & 0 & 0 \\
   0 & 0 & 0 & 0 & 0 & 0 & 0 & 0 & 0 & 0 & 1 & 0 & 0 & 0 & 0 & 0 & 0 \\
   0 & 0 & 0 & 0 & 0 & 0 & 0 & 0 & 0 & 0 & 0 & 1 & 0 & 0 & 0 & 0 & 0 \\
   0 & 0 & 0 & 0 & 0 & 0 & 0 & 0 & 0 & 0 & 0 & 0 & 1 & 0 & 0 & 0 & 0 \\
   0 & 0 & 0 & 0 & 0 & 0 & 0 & 0 & 0 & 0 & 0 & 0 & 0 & 1 & 0 & 0 & 0 \\
   0 & 0 & 0 & 0 & 0 & 0 & 0 & 0 & 0 & 0 & 0 & 0 & 0 & 0 & 1 & 0 & 0 \\
   0 & 0 & 0 & 0 & 0 & 0 & 0 & 0 & 0 & 0 & 0 & 0 & 0 & 0 & 0 & 1 & 0 \\
   0 & 0 & 0 & 0 & 0 & 0 & 0 & 0 & 0 & 0 & 0 & 0 & 0 & 0 & 0 & 0 & 1 \\
   \hline
\color{blue}{\textbf{1}} & \color{blue}{\textbf{0}} & \color{blue}{\textbf{0}} & \color{blue}{\textbf{0}} & \color{blue}{\textbf{1}} & \color{blue}{\textbf{0}} & \color{blue}{\textbf{0}} & \color{blue}{\textbf{0}} & \color{blue}{\textbf{1}} & \color{blue}{\textbf{0}} & \color{blue}{\textbf{0}} & \color{blue}{\textbf{0}} & \color{blue}{\textbf{1}} & \color{blue}{\textbf{0}} & \color{blue}{\textbf{0}} & \color{blue}{\textbf{0}} & |\color{red}{\textbf{1}} \\
  \color{blue}{\textbf{0}} & \color{blue}{\textbf{1}} & \color{blue}{\textbf{0}} & \color{blue}{\textbf{0}} & \color{blue}{\textbf{0}} & \color{blue}{\textbf{1}} & \color{blue}{\textbf{0}} & \color{blue}{\textbf{0}} & \color{blue}{\textbf{0}} & \color{blue}{\textbf{1}} & \color{blue}{\textbf{0}} & \color{blue}{\textbf{0}} & \color{blue}{\textbf{0}} & \color{blue}{\textbf{1}} & \color{blue}{\textbf{0}} & \color{blue}{\textbf{0}} & |\color{red}{\textbf{1}} \\
  \color{blue}{\textbf{0}} & \color{blue}{\textbf{0}} & \color{blue}{\textbf{1}} & \color{blue}{\textbf{0}} & \color{blue}{\textbf{0}} & \color{blue}{\textbf{0}} & \color{blue}{\textbf{1}} & \color{blue}{\textbf{0}} & \color{blue}{\textbf{0}} & \color{blue}{\textbf{0}} & \color{blue}{\textbf{1}} & \color{blue}{\textbf{0}} & \color{blue}{\textbf{0}} & \color{blue}{\textbf{0}} & \color{blue}{\textbf{1}} & \color{blue}{\textbf{0}} & |\color{red}{\textbf{1}} \\
  \color{blue}{\textbf{0}} & \color{blue}{\textbf{0}} & \color{blue}{\textbf{0}} & \color{blue}{\textbf{1}} & \color{blue}{\textbf{0}} & \color{blue}{\textbf{0}} & \color{blue}{\textbf{0}} & \color{blue}{\textbf{1}} & \color{blue}{\textbf{0}} & \color{blue}{\textbf{0}} & \color{blue}{\textbf{0}} & \color{blue}{\textbf{1}} & \color{blue}{\textbf{0}} & \color{blue}{\textbf{0}} & \color{blue}{\textbf{0}} & \color{blue}{\textbf{1}}& |\color{red}{\textbf{1}} \\
   \end{array}\right]$$
}
\end{small}
The scalar index code is
\begin{align*}
\mathfrak{C}=\{&x_1+x_{18},~~~~ x_2+x_{19}, ~~~~ x_3+x_{20},\\& x_4+x_{21},
~~~~ x_5+x_{18}, ~~~~ x_6+x_{19}, \\&x_7+x_{20}, ~~~~ x_8+x_{21}, ~~~~ x_9+x_{18}, \\& x_{10}+x_{19}, ~~~ x_{11}+x_{20}, ~~~ x_{12}+x_{21},\\&
x_{13}+x_{18}, ~~~ x_{14}+x_{19}, ~~~ x_{15}+x_{20}, \\& x_{16}+x_{21},
~~~ x_{17}+x_{18}+x_{19}+x_{20}+x_{21}\}.
\end{align*}

The receivers decode their wanted message by using the code symbols given in Table \ref{table1} and their antidotes.
\begin{table}[h]
\centering
\setlength\extrarowheight{3.5pt}
\begin{small}
\begin{tabular}{|c|c|c|c|}
\hline
\textbf{Rx} & {$\mathcal{W}_k$} & \textbf{Code symbols used to decode} $\mathcal{W}_k$\\
\hline
\textbf{$R_1$} & $x_1$ & $x_1+x_{18},~x_5+x_{18}$ \\
\hline
\textbf{$R_2$} & $x_2$& $x_2+x_{19},~x_6+x_{19}$ \\
\hline 
\textbf{$R_3$} & $x_3$ & $x_3+x_{20},~x_7+x_{20}$  \\
\hline
\textbf{$R_4$} & $x_4$ & $x_4+x_{21},~x_8+x_{21}$  \\
\hline
\textbf{$R_5$} & $x_5$ & $x_5+x_{18},~x_9+x_{18}$  \\
\hline
\textbf{$R_6$} & $x_6$ & $x_6+x_{19},~x_{10}+x_{19}$ \\
\hline
\textbf{$R_7$} & $x_7$ & $x_7+x_{20},~x_{11}+x_{20}$ \\
\hline 
\textbf{$R_8$} & $x_8$ & $x_8+x_{21},~x_{12}+x_{21}$ \\
\hline
\textbf{$R_9$} & $x_9$ & $x_9+x_{18},~x_{13}+x_{18}$ \\
\hline
\textbf{$R_{10}$} &$x_{10}$ & $x_{10}+x_{19},~x_{14}+x_{19}$ \\
\hline
\textbf{$R_{11}$} & $x_{11}$ & $x_{11}+x_{20},~x_{15}+x_{20}$ \\
\hline 
\textbf{$R_{12}$} & $x_{12}$ & $x_{12}+x_{21},~x_{16}+x_{21}$ \\
\hline
\textbf{$R_{13}$} & $x_{13}$ &  $x_{13}+x_{18},~x_{14}+x_{19},~x_{15}+x_{20}$,\\
~&~&$x_{16}+x_{21},~x_{17}+x_{18}+x_{19}+x_{20}+x_{21}$ \\
\hline
\textbf{$R_{14}$} & $x_{14}$ & $x_{14}+x_{19},~x_{15}+x_{20},~x_{16}+x_{21}$, \\
~&~&$x_{17}+x_{18}+x_{19}+x_{20}+x_{21}$ \\
\hline
\textbf{$R_{15}$} & $x_{15}$ & $x_{15}+x_{20},~x_{16}+x_{21},$\\
~&~&$x_{17}+x_{18}+x_{19}+x_{20}+x_{21}$ \\
\hline 
\textbf{$R_{16}$} & $x_{16}$ & $x_{16}+x_{21},x_{17}+x_{18}+x_{19}+x_{20}+x_{21}$ \\
\hline
\textbf{$R_{17}$} & $x_{17}$ & $x_{17}+x_{18}+x_{19}+x_{20}+x_{21}$ \\
\hline
\textbf{$R_{18}$} & $x_{18}$ &  $x_1+x_{18}$\\
\hline
\textbf{$R_{19}$} & $x_{19}$ & $x_2+x_{19}$ \\
\hline
\textbf{$R_{20}$} & $x_{20}$ &  $x_3+x_{20}$ \\
\hline
\textbf{$R_{21}$} & $x_{21}$ & $x_4+x_{21}$ \\
\hline 
\end{tabular}
\end{small}
\vspace{5pt}
\caption{}
\label{table1}
\end{table}
\end{example}
\begin{example}
Consider a one-sided neighboring antidote MUICP with $K=21, D=17$. 

For this index coding problem, we have $\lambda_1=4,\lambda_2=1$ and $l=2$.

An encoding matrix $\mathbf{L}_{21 \times 4}$ for this index coding problem is obtained in the following steps $\mathbf{C}_{4 \times 1}^{\mathsf{T}}\triangleq\mathbf{D}_{1 \times 4}^{(1)}  
 \rightarrow\mathbf{D}_{17 \times 4}^{(2)}  \rightarrow \mathbf{L}_{21 \times 4}$ by concatenating suitable rectangular circulant matrices.
The encoding matrix $\mathbf{L}_{21 \times 4}$ is given below.
\begin{small}
\arraycolsep=3pt
\setlength\extrarowheight{-2.0pt}
{
$$L_{21 \times 4}=\left[\begin{array}{*{20}c}
   1 & 0 & 0 & 0 \\
   0 & 1 & 0 & 0 \\
   0 & 0 & 1 & 0 \\
   0 & 0 & 0 & 1 \\
 \hline
 \color{blue}{\textbf{1}} & \color{blue}{\textbf{0}} & \color{blue}{\textbf{0}} & \color{blue}{\textbf{0}} \\
  \color{blue}{\textbf{0}} & \color{blue}{\textbf{1}} & \color{blue}{\textbf{0}} & \color{blue}{\textbf{0}} \\
  \color{blue}{\textbf{0}} & \color{blue}{\textbf{0}} & \color{blue}{\textbf{1}} & \color{blue}{\textbf{0}} \\
  \color{blue}{\textbf{0}} & \color{blue}{\textbf{0}} & \color{blue}{\textbf{0}} & \color{blue}{\textbf{1}}\\
    \color{blue}{\textbf{1}} & \color{blue}{\textbf{0}} & \color{blue}{\textbf{0}} & \color{blue}{\textbf{0}} \\
  \color{blue}{\textbf{0}} & \color{blue}{\textbf{1}} & \color{blue}{\textbf{0}} & \color{blue}{\textbf{0}} \\
  \color{blue}{\textbf{0}} & \color{blue}{\textbf{0}} & \color{blue}{\textbf{1}} & \color{blue}{\textbf{0}} \\
  \color{blue}{\textbf{0}} & \color{blue}{\textbf{0}} & \color{blue}{\textbf{0}} & \color{blue}{\textbf{1}}\\
  \color{blue}{\textbf{1}} & \color{blue}{\textbf{0}} & \color{blue}{\textbf{0}} & \color{blue}{\textbf{0}} \\
  \color{blue}{\textbf{0}} & \color{blue}{\textbf{1}} & \color{blue}{\textbf{0}} & \color{blue}{\textbf{0}} \\
  \color{blue}{\textbf{0}} & \color{blue}{\textbf{0}} & \color{blue}{\textbf{1}} & \color{blue}{\textbf{0}} \\
  \color{blue}{\textbf{0}} & \color{blue}{\textbf{0}} & \color{blue}{\textbf{0}} & \color{blue}{\textbf{1}}\\
  \color{blue}{\textbf{1}} & \color{blue}{\textbf{0}} & \color{blue}{\textbf{0}} & \color{blue}{\textbf{0}} \\
  \color{blue}{\textbf{0}} & \color{blue}{\textbf{1}} & \color{blue}{\textbf{0}} & \color{blue}{\textbf{0}} \\
  \color{blue}{\textbf{0}} & \color{blue}{\textbf{0}} & \color{blue}{\textbf{1}} & \color{blue}{\textbf{0}} \\
  \color{blue}{\textbf{0}} & \color{blue}{\textbf{0}} & \color{blue}{\textbf{0}} & \color{blue}{\textbf{1}}\\
  \color{red}{\textbf{1}} & \color{red}{\textbf{1}} & \color{red}{\textbf{1}} & \color{red}{\textbf{1}} \\
   \end{array}\right]$$
}
\end{small}
The scalar index code is
\begin{align*}
\mathfrak{C}=\{&x_1+x_{5}+x_9+x_{13}+x_{17}+x_{21},\\& x_{2}+x_6+x_{10}+x_{14}+x_{18}+x_{21},\\& 
x_{3}+x_7+x_{11}+x_{15}+x_{19}+x_{21}, \\&
x_4+x_{8}+x_{12}+x_{16}+x_{20}+x_{21}\}.
\end{align*}

The receivers decode their wanted message by using the code symbols given in Table \ref{table2} and their antidotes.
\begin{table}[h]
\centering
\setlength\extrarowheight{2.2pt}
\begin{small}
\begin{tabular}{|c|c|c|c|c|}
\hline
\textbf{Rx} & {$\mathcal{W}_k$} & \textbf{Code symbols used to decode} $\mathcal{W}_k$\\
\hline
\textbf{$R_1$} & $x_1$ & $x_1+x_{5}+x_9+x_{13}+x_{17}+x_{21},$ \\
~ & ~ & $x_{2}+x_6+x_{10}+x_{14}+x_{18}+x_{21}$ \\
\hline
\textbf{$R_2$} & $x_2$& $x_{2}+x_6+x_{10}+x_{14}+x_{18}+x_{21}$,\\ 
~ &~&$x_{3}+x_7+x_{11}+x_{15}+x_{19}+x_{21}$ \\
\hline 
\textbf{$R_3$} & $x_3$ & $x_{3}+x_7+x_{11}+x_{15}+x_{19}+x_{21}$, \\
~ & ~& $x_4+x_{8}+x_{12}+x_{16}+x_{20}+x_{21}$\\
\hline
\textbf{$R_4$} & $x_4$ & $x_4+x_{8}+x_{12}+x_{16}+x_{20}+x_{21}$  \\
\hline
\textbf{$R_5$} & $x_5$ & $x_1+x_{5}+x_9+x_{13}+x_{17}+x_{21}$  \\
\hline
\textbf{$R_6$} & $x_6$ & $x_{2}+x_6+x_{10}+x_{14}+x_{18}+x_{21}$  \\
\hline
\textbf{$R_7$} & $x_7$ & $x_{3}+x_7+x_{11}+x_{15}+x_{19}+x_{21}$  \\
\hline 
\textbf{$R_8$} & $x_8$ & $x_4+x_{8}+x_{12}+x_{16}+x_{20}+x_{21}$  \\
\hline
\textbf{$R_9$} & $x_9$ & $x_1+x_{5}+x_9+x_{13}+x_{17}+x_{21}$  \\
\hline
\textbf{$R_{10}$} &$x_{10}$ & $x_{2}+x_6+x_{10}+x_{14}+x_{18}+x_{21}$  \\
\hline
\textbf{$R_{11}$} & $x_{11}$ & $x_{3}+x_7+x_{11}+x_{15}+x_{19}+x_{21}$  \\
\hline 
\textbf{$R_{12}$} & $x_{12}$ & $x_4+x_{8}+x_{12}+x_{16}+x_{20}+x_{21}$  \\
\hline
\textbf{$R_{13}$} & $x_{13}$ & $x_1+x_{5}+x_9+x_{13}+x_{17}+x_{21}$  \\
\hline
\textbf{$R_{14}$} & $x_{14}$ & $x_{2}+x_6+x_{10}+x_{14}+x_{18}+x_{21}$  \\
\hline
\textbf{$R_{15}$} & $x_{15}$& $x_{3}+x_7+x_{11}+x_{15}+x_{19}+x_{21}$  \\
\hline 
\textbf{$R_{16}$} & $x_{16}$ &$x_4+x_{8}+x_{12}+x_{16}+x_{20}+x_{21}$  \\
\hline
\textbf{$R_{17}$} & $x_{17}$ & $x_1+x_{5}+x_9+x_{13}+x_{17}+x_{21}$  \\
\hline
\textbf{$R_{18}$} & $x_{18}$ &  $x_{2}+x_6+x_{10}+x_{14}+x_{18}+x_{21}$ \\
\hline
\textbf{$R_{19}$} & $x_{19}$ & $x_{3}+x_7+x_{11}+x_{15}+x_{19}+x_{21}$ \\
\hline
\textbf{$R_{20}$} & $x_{20}$ &  $x_4+x_{8}+x_{12}+x_{16}+x_{20}+x_{21}$  \\
\hline
\textbf{$R_{21}$} & $x_{21}$ & $x_1+x_{5}+x_9+x_{13}+x_{17}+x_{21}$  \\
\hline 
\end{tabular}
\end{small}
\vspace{5pt}
\caption{}
\label{table2}
\end{table}
\end{example}

\begin{example}
\label{ex3}
Consider a one-sided neighboring antidote MUICP with $K=44, D=17$. 

For this index coding problem, we have $\lambda_1=10,\lambda_2=7,\lambda_3=3,\lambda_4=1$ and $l=4$.

An encoding matrix $\mathbf{L}_{44 \times 27}$ for this index coding problem is obtained in the following steps $\mathbf{C}_{3 \times 1}^{\mathsf{T}}\triangleq\mathbf{D}_{1 \times 3}^{(1)} 
 \rightarrow \mathbf{D}_{7 \times 3}^{(2)} \rightarrow \mathbf{D}_{7 \times 10}^{(3)} \rightarrow \mathbf{D}_{17 \times 10}^{(4)} \rightarrow \mathbf{D}_{17 \times 27}^{(5)} \rightarrow \mathbf{D}_{44 \times 27}=\mathbf{L}_{44 \times 27}$ by concatenating suitable rectangular circulant matrices.
The encoding matrix $\mathbf{L}_{44 \times 27}$ is given in Fig. \ref{ex3matrix}.

\begin{figure*}
\arraycolsep=1.3pt
\setlength\extrarowheight{-1.0pt}
\begin{small}
{
$$L_{44 \times 27}=\left[
\begin{array}{ccccccccccccccccccccccccccc}
1 & 0 & 0 & 0 & 0 & 0 & 0 & 0 & 0 & 0 & 0 & 0& 0 & 0 & 0 & 0 & 0 & 0 & 0 & 0 & 0 & 0 & 0 & 0 & 0 & 0 & 0\\
0 & 1 & 0 & 0 & 0 & 0 & 0 & 0 & 0 & 0 & 0 & 0& 0 & 0 & 0 & 0 & 0 & 0 & 0 & 0 & 0 & 0 & 0 & 0 & 0 & 0 & 0\\
0 & 0 & 1 & 0 & 0 & 0 & 0 & 0 & 0 & 0 & 0 & 0& 0 & 0 & 0 & 0 & 0 & 0 & 0 & 0 & 0 & 0 & 0 & 0 & 0 & 0 & 0\\
0 & 0 & 0 & 1 & 0 & 0 & 0 & 0 & 0 & 0 & 0 & 0& 0 & 0 & 0 & 0 & 0 & 0 & 0 & 0 & 0 & 0 & 0 & 0 & 0 & 0 & 0\\
0 & 0 & 0 & 0 & 1 & 0 & 0 & 0 & 0 & 0 & 0 & 0& 0 & 0 & 0 & 0 & 0 & 0 & 0 & 0 & 0 & 0 & 0 & 0 & 0 & 0 & 0\\
0 & 0 & 0 & 0 & 0 & 1 & 0 & 0 & 0 & 0 & 0 & 0& 0 & 0 & 0 & 0 & 0 & 0 & 0 & 0 & 0 & 0 & 0 & 0 & 0 & 0 & 0\\
 0 & 0 & 0 & 0 & 0  & 0 & 1 & 0 & 0 & 0 & 0& 0 & 0 & 0 & 0 & 0 & 0 & 0 & 0 & 0 & 0 & 0 & 0 & 0 & 0 & 0 & 0\\
 0 & 0 & 0 & 0 & 0 & 0 & 0 & 1 & 0 & 0 & 0 & 0& 0 & 0 & 0 & 0 & 0 & 0 & 0 & 0 & 0 & 0 & 0 & 0 & 0 & 0 & 0\\
 0 & 0 & 0 & 0 & 0 & 0 & 0 & 0 & 1 & 0 & 0 & 0& 0 & 0 & 0 & 0 & 0 & 0 & 0 & 0 & 0 & 0 & 0 & 0 & 0 & 0 & 0\\
0 & 0 & 0 & 0 & 0 & 0 & 0 & 0 & 0 & 1 & 0 & 0& 0 & 0 & 0 & 0 & 0 & 0 & 0 & 0 & 0 & 0 & 0 & 0 & 0 & 0 & 0\\
 0 & 0 & 0 & 0 & 0 & 0 & 0 & 0 & 0 & 0 & 1 & 0& 0 & 0 & 0 & 0 & 0 & 0 & 0 & 0 & 0 & 0 & 0 & 0 & 0 & 0 & 0\\
 0 & 0 & 0 & 0 & 0 & 0 & 0 & 0 & 0 & 0 & 0 & 1& 0 & 0 & 0 & 0 & 0 & 0 & 0 & 0 & 0 & 0 & 0 & 0 & 0 & 0 & 0\\
 0 & 0 & 0 & 0 & 0 & 0 & 0 & 0 & 0 & 0 & 0 & 0& 1 & 0 & 0 & 0 & 0 & 0 & 0 & 0 & 0 & 0 & 0 & 0 & 0 & 0 & 0\\
0 & 0 & 0 & 0 & 0 & 0 & 0 & 0 & 0 & 0 & 0 & 0& 0 & 1 & 0 & 0 & 0 & 0 & 0 & 0 & 0 & 0 & 0 & 0 & 0 & 0 & 0\\
0 & 0 & 0 & 0 & 0 & 0 & 0 & 0 & 0 & 0 & 0 & 0& 0 & 0 & 1 & 0 & 0 & 0 & 0 & 0 & 0 & 0 & 0 & 0 & 0 & 0 & 0\\
 0 & 0 & 0 & 0 & 0 & 0 & 0 & 0 & 0 & 0 & 0 & 0& 0 & 0 & 0 & 1 & 0 & 0 & 0 & 0 & 0 & 0 & 0 & 0 & 0 & 0 & 0\\
0 & 0 & 0 & 0 & 0 & 0 & 0 & 0 & 0 & 0 & 0 & 0& 0 & 0 & 0 & 0 & 1 & 0 & 0 & 0 & 0 & 0 & 0 & 0 & 0 & 0 & 0 \\
 0 & 0 & 0 & 0 & 0 & 0 & 0 & 0 & 0 & 0 & 0 & 0& 0 & 0 & 0 & 0 & 0 & 1 & 0 & 0 & 0 & 0 & 0 & 0 & 0 & 0 & 0\\
 0 & 0 & 0 & 0 & 0 & 0 & 0 & 0 & 0 & 0 & 0 & 0& 0 & 0 & 0 & 0 & 0 & 0 &1 & 0 & 0 & 0 & 0 & 0 & 0 & 0 & 0\\
  0 & 0 & 0 & 0 & 0 & 0 & 0 & 0 & 0 & 0 & 0 & 0& 0 & 0 & 0 & 0 & 0 & 0 & 0 & 1 & 0 & 0 & 0 & 0 & 0 & 0 & 0\\
 0 & 0 & 0 & 0 & 0 & 0 & 0 & 0 & 0& 0  & 0 & 0 & 0& 0 & 0 & 0 & 0 & 0 & 0 & 0 & 1 & 0 & 0 & 0 & 0 & 0 & 0\\
 0 & 0 & 0 & 0 & 0 & 0 & 0 & 0 & 0 & 0 & 0 & 0 & 0& 0 & 0 & 0 & 0 & 0 & 0 & 0 & 0 & 1 & 0 & 0 & 0 & 0 & 0\\
  0 & 0 & 0 & 0 & 0 & 0 & 0 & 0 & 0 & 0 & 0 & 0& 0 & 0 & 0 & 0 & 0 & 0 & 0 & 0 & 0 & 0 & 1 & 0 & 0 & 0 & 0\\
 0 & 0 & 0 & 0 & 0 & 0 & 0 & 0 & 0 & 0 & 0 & 0& 0 & 0 & 0 & 0 & 0 & 0 & 0 & 0 & 0 & 0 & 0 & 1 & 0 & 0 & 0\\
 0 & 0 & 0 & 0 & 0 & 0 & 0 & 0 & 0 & 0 & 0 & 0& 0 & 0 & 0 & 0 & 0 & 0 & 0 & 0 & 0 & 0 & 0 & 0 & 1 & 0 & 0\\
  0 & 0 & 0 & 0 & 0 & 0 & 0 & 0 & 0 & 0 & 0 & 0& 0 & 0 & 0 & 0 & 0 & 0 & 0 & 0 & 0 & 0 & 0 & 0 & 0 & 1 & 0\\
  0 & 0 & 0 & 0 & 0 & 0 & 0 & 0 & 0 & 0 & 0 & 0& 0 & 0 & 0 & 0 & 0 & 0 & 0 & 0 & 0 & 0 & 0 & 0 & 0 & 0 & 1\\
\hline
1 & 0 & 0 & 0 & 0 & 0 & 0 & 0 & 0 & 0 & 0 & 0& 0 & 0 & 0 & 0 & 0 & |{\color{red}\textbf{1}} & {\color{red}\textbf{0}} & {\color{red}\textbf{0}} & {\color{red}\textbf{0}} & {\color{red}\textbf{0}} & {\color{red}\textbf{0}} & {\color{red}\textbf{0}} & {\color{red}\textbf{0}} & {\color{red}\textbf{0}} & {\color{red}\textbf{0}}\\
  0 & 1 & 0 & 0 & 0 & 0 & 0 & 0 & 0 & 0 & 0 & 0& 0 & 0 & 0 & 0 & 0 & |{\color{red}\textbf{0}} & {\color{red}\textbf{1}} & {\color{red}\textbf{0}} &{\color{red}\textbf{0}} & {\color{red}\textbf{0}} & {\color{red}\textbf{0}} & {\color{red}\textbf{0}} & {\color{red}\textbf{0}} & {\color{red}\textbf{0}} & {\color{red}\textbf{0}}\\
  0 & 0 & 1 & 0 & 0 & 0 & 0 & 0 & 0 & 0 & 0 & 0& 0 & 0 & 0 & 0 & 0 & |{\color{red}\textbf{0}} & {\color{red}\textbf{0}} &  {\color{red}\textbf{1}} & {\color{red}\textbf{0}}& {\color{red}\textbf{0}} & {\color{red}\textbf{0}} & {\color{red}\textbf{0}} & {\color{red}\textbf{0}} & {\color{red}\textbf{0}} & {\color{red}\textbf{0}}\\
  0 & 0 & 0 & 1 & 0 & 0 & 0 & 0 & 0 & 0 & 0 & 0& 0 & 0 & 0 & 0 & 0 & |{\color{red}\textbf{0}} & {\color{red}\textbf{0}} & {\color{red}\textbf{0}} &  {\color{red}\textbf{1}} & {\color{red}\textbf{0}} & {\color{red}\textbf{0}}& {\color{red}\textbf{0}} & {\color{red}\textbf{0}}& {\color{red}\textbf{0}} & {\color{red}\textbf{0}}\\
 0 & 0 & 0 & 0 & 1 & 0 & 0 & 0 & 0 & 0 & 0 & 0& 0 & 0 & 0 & 0 & 0 & |{\color{red}\textbf{0}} & {\color{red}\textbf{0}} &{\color{red}\textbf{ 0}} & {\color{red}\textbf{0}} & {\color{red}\textbf{1}}& {\color{red}\textbf{0}} & {\color{red}\textbf{0}} & {\color{red}\textbf{0}} & {\color{red}\textbf{0}} & {\color{red}\textbf{0}}\\
  0 & 0 & 0 & 0 & 0 & 1 & 0 & 0 & 0 & 0 & 0 & 0& 0 & 0 & 0 & 0 & 0 & |{\color{red}\textbf{0}} & {\color{red}\textbf{0}} & {\color{red}\textbf{0}} & {\color{red}\textbf{0}} & {\color{red}\textbf{0}} & {\color{red}\textbf{1}} & {\color{red}\textbf{0}} & {\color{red}\textbf{0}} & {\color{red}\textbf{0}} & {\color{red}\textbf{0}}\\
  0 & 0 & 0 & 0 & 0  & 0 & 1 & 0 & 0 & 0 & 0& 0 & 0 & 0 & 0 & 0 & 0 & |{\color{red}\textbf{0}} & {\color{red}\textbf{0}} & {\color{red}\textbf{0}} & {\color{red}\textbf{0}} & {\color{red}\textbf{0}} & {\color{red}\textbf{0}} & {\color{red}\textbf{1}} & {\color{red}\textbf{0}} & {\color{red}\textbf{0}} & {\color{red}\textbf{0}} \\
 0 & 0 & 0 & 0 & 0 & 0 & 0 & 1 & 0 & 0 & 0 & 0& 0 & 0 & 0 & 0 & 0 & |{\color{red}\textbf{0}} & {\color{red}\textbf{0}} & {\color{red}\textbf{0}} & {\color{red}\textbf{0}} & {\color{red}\textbf{0}} & {\color{red}\textbf{0}} & {\color{red}\textbf{0}} & {\color{red}\textbf{1}} & {\color{red}\textbf{0}} & {\color{red}\textbf{0}}\\
 0 & 0 & 0 & 0 & 0 & 0 & 0 & 0 & 1 & 0 & 0 & 0& 0 & 0 & 0 & 0 & 0 & |{\color{red}\textbf{0}} & {\color{red}\textbf{0}} & {\color{red}\textbf{0}} & {\color{red}\textbf{0}} & {\color{red}\textbf{0}} & {\color{red}\textbf{0}} & {\color{red}\textbf{0}} & {\color{red}\textbf{0}} & {\color{red}\textbf{1}} & {\color{red}\textbf{0}}\\
 0 & 0 & 0 & 0 & 0 & 0 & 0 & 0 & 0 & 1 & 0 & 0& 0 & 0 & 0 & 0 & 0 & |\underline{{\color{red}\textbf{0}}} & \underline{{\color{red}\textbf{0}}} & \underline{{\color{red}\textbf{0}}} & \underline{{\color{red}\textbf{0}}} & \underline{{\color{red}\textbf{0}}} & \underline{{\color{red}\textbf{0}}} & \underline{{\color{red}\textbf{0}}} & \underline{{\color{red}\textbf{0}}} & \underline{{\color{red}\textbf{0}}} & \underline{{\color{red}\textbf{1}}}\\

 0 & 0 & 0 & 0 & 0 & 0 & 0 & 0 & 0 & 0 & 1 & 0& 0 & 0 & 0 & 0 & 0 & |{\color{green}\textbf{1}} & {\color{green} \textbf{0}} & {\color{green}\textbf{0}} & {\color{green}\textbf{0}} & {\color{green}\textbf{0}} & {\color{green}\textbf{0}} & {\color{green}\textbf{0}} &|{\color{blue}\textbf{1}} & {\color{blue}\textbf{0}} & {\color{blue}\textbf{0}}\\
 0 & 0 & 0 & 0 & 0 & 0 & 0 & 0 & 0 & 0 & 0 & 1& 0 & 0 & 0 & 0 & 0 & |{\color{green}\textbf{0}} &{\color{green}\textbf{1}} & {\color{green}\textbf{0}} & {\color{green}\textbf{0}} & {\color{green}\textbf{0}} & {\color{green}\textbf{0}} & {\color{green}\textbf{0}} & |{\color{blue}\textbf{0}} & {\color{blue}\textbf{1}} & {\color{blue}\textbf{0}}\\
 0 & 0 & 0 & 0 & 0 & 0 & 0 & 0 & 0 & 0 & 0 & 0& 1 & 0 & 0 & 0 & 0 & |{\color{green}\textbf{0}} & {\color{green}\textbf{0}} & {\color{green}\textbf{1}} & {\color{green}\textbf{0}}  & {\color{green}\textbf{0}} & {\color{green}\textbf{0}} & {\color{green}\textbf{0}} & |{\color{blue}\textbf{0}} & {\color{blue}\textbf{0}} & {\color{blue}\textbf{1}}\\
 0 & 0 & 0 & 0 & 0 & 0 & 0 & 0 & 0 & 0 & 0 & 0& 0 & 1 & 0 & 0 & 0 & |{\color{green}\textbf{0}} & {\color{green}\textbf{0}} & {\color{green}\textbf{0}} & {\color{green}\textbf{1}} & {\color{green}\textbf{0}} & {\color{green}\textbf{0}} & {\color{green}\textbf{0}} & |{\color{blue}\textbf{1}} & {\color{blue}\textbf{0}} & {\color{blue}\textbf{0}}\\
0 & 0 & 0 & 0 & 0 & 0 & 0 & 0 & 0 & 0 & 0 & 0& 0 & 0 & 1 & 0 & 0 & |{\color{green}\textbf{0}} & {\color{green}\textbf{0}} & {\color{green}\textbf{0}} & {\color{green}\textbf{0}} & {\color{green}\textbf{1}} & {\color{green}\textbf{0}} & {\color{green}\textbf{0}} & |{\color{blue}\textbf{0}} & {\color{blue}\textbf{1}} & {\color{blue}\textbf{0}}\\
0 & 0 & 0 & 0 & 0 & 0 & 0 & 0 & 0 & 0 & 0 & 0& 0 & 0 & 0 & 1 & 0 & |{\color{green}\textbf{0}} & {\color{green}\textbf{0}} & {\color{green}\textbf{0}} & {\color{green}\textbf{0}} & {\color{green}\textbf{0}} & {\color{green}\textbf{1}} & {\color{green}\textbf{0}} & |\underline{{\color{blue}\textbf{0}}} & \underline{{\color{blue}\textbf{0}}} & \underline{{\color{blue}\textbf{1}}}\\
0 & 0 & 0 & 0 & 0 & 0 & 0 & 0 & 0 & 0 & 0 & 0& 0 & 0 & 0 & 0 & 1 & |{\color{green}\textbf{0}} & {\color{green}\textbf{0}} & {\color{green}\textbf{0}} & {\color{green}\textbf{0}} & {\color{green}\textbf{0}} & {\color{green}\textbf{0}} & {\color{green}\textbf{1}} & |\textbf{1} & \textbf{1} & \textbf{1}\\
 \end{array}
\right]$$
}
\end{small}
\caption{Encoding matrix for the MUICP in Example \ref{ex3}.}
\label{ex3matrix}
\end{figure*}
\end{example}
\begin{example}
\label{ex4}
Consider a one-sided neighboring antidote MUICP with $K=44, D=27$. 

For this index coding problem, we have $\lambda_1=17,\lambda_2=10,\lambda_3=7,\lambda_4=3,\lambda_5=1$ and $l=5$.

The encoding matrix $\mathbf{L}_{44 \times 17}$ for this index coding problem is obtained in the following steps $\mathbf{C}_{3 \times 1}\triangleq\mathbf{D}_{3 \times 1}^{(1)} 
 \rightarrow \mathbf{D}_{3 \times 7}^{(2)} \rightarrow \mathbf{D}_{10 \times 7}^{(3)} \rightarrow \mathbf{D}_{10 \times 17}^{(4)} \rightarrow \mathbf{L}_{44 \times 17} $ by concatenating suitable rectangular circulant matrices.
The encoding matrix $\mathbf{L}_{44 \times 17}$ is given in Fig. \ref{ex4matrix}.

\begin{figure*}
\arraycolsep=1.6pt
\setlength\extrarowheight{-3.0pt}
\begin{small}
{
$$L_{44 \times 17}=\left[
\begin{array}{ccccccccccccccccc}
1 & 0 & 0 & 0 & 0 & 0 & 0 & 0 & 0 & 0 & 0 & 0& 0 & 0 & 0 & 0 & 0\\
0 & 1 & 0 & 0 & 0 & 0 & 0 & 0 & 0 & 0 & 0 & 0& 0 & 0 & 0 & 0 & 0\\
0 & 0 & 1 & 0 & 0 & 0 & 0 & 0 & 0 & 0 & 0 & 0& 0 & 0 & 0 & 0 & 0\\
0 & 0 & 0 & 1 & 0 & 0 & 0 & 0 & 0 & 0 & 0 & 0& 0 & 0 & 0 & 0 & 0\\
0 & 0 & 0 & 0 & 1 & 0 & 0 & 0 & 0 & 0 & 0 & 0& 0 & 0 & 0 & 0 & 0\\
0 & 0 & 0 & 0 & 0 & 1 & 0 & 0 & 0 & 0 & 0 & 0& 0 & 0 & 0 & 0 & 0\\
0 & 0 & 0 & 0 & 0 & 0 & 1 & 0 & 0 & 0 & 0& 0 & 0 & 0 & 0 & 0 & 0\\
0 & 0 & 0 & 0 & 0 & 0 & 0 & 1 & 0 & 0 & 0 & 0& 0 & 0 & 0 & 0 & 0\\
0 & 0 & 0 & 0 & 0 & 0 & 0 & 0 & 1 & 0 & 0 & 0& 0 & 0 & 0 & 0 & 0\\
0 & 0 & 0 & 0 & 0 & 0 & 0 & 0 & 0 & 1 & 0 & 0& 0 & 0 & 0 & 0 & 0\\
0 & 0 & 0 & 0 & 0 & 0 & 0 & 0 & 0 & 0 & 1 & 0& 0 & 0 & 0 & 0 & 0\\
0 & 0 & 0 & 0 & 0 & 0 & 0 & 0 & 0 & 0 & 0 & 1& 0 & 0 & 0 & 0 & 0\\
0 & 0 & 0 & 0 & 0 & 0 & 0 & 0 & 0 & 0 & 0 & 0& 1 & 0 & 0 & 0 & 0\\
0 & 0 & 0 & 0 & 0 & 0 & 0 & 0 & 0 & 0 & 0 & 0& 0 & 1 & 0 & 0 & 0\\
0 & 0 & 0 & 0 & 0 & 0 & 0 & 0 & 0 & 0 & 0 & 0& 0 & 0 & 1 & 0 & 0\\
0 & 0 & 0 & 0 & 0 & 0 & 0 & 0 & 0 & 0 & 0 & 0& 0 & 0 & 0 & 1 & 0\\
0 & 0 & 0 & 0 & 0 & 0 & 0 & 0 & 0 & 0 & 0 & 0& 0 & 0 & 0 & 0 & 1\\
\hline
1 & 0 & 0 & 0 & 0 & 0 & 0 & 0 & 0 & 0 & 0 & 0& 0 & 0 & 0 & 0 & 0\\
0 & 1 & 0 & 0 & 0 & 0 & 0 & 0 & 0 & 0 & 0 & 0& 0 & 0 & 0 & 0 & 0\\
0 & 0 & 1 & 0 & 0 & 0 & 0 & 0 & 0 & 0 & 0 & 0& 0 & 0 & 0 & 0 & 0\\
0 & 0 & 0 & 1 & 0 & 0 & 0 & 0 & 0& 0  & 0 & 0 & 0& 0 & 0 & 0 & 0\\
0 & 0 & 0 & 0 & 1 & 0 & 0 & 0 & 0 & 0 & 0 & 0 & 0& 0 & 0 & 0 & 0\\
0 & 0 & 0 & 0 & 0 & 1 & 0 & 0 & 0 & 0 & 0 & 0& 0 & 0 & 0 & 0 & 0\\
0 & 0 & 0 & 0 & 0 & 0 & 1 & 0 & 0 & 0 & 0 & 0& 0 & 0 & 0 & 0 & 0\\
0 & 0 & 0 & 0 & 0 & 0 & 0 & 1 & 0 & 0 & 0 & 0& 0 & 0 & 0 & 0 & 0\\
0 & 0 & 0 & 0 & 0 & 0 & 0 & 0 & 1 & 0 & 0 & 0& 0 & 0 & 0 & 0 & 0\\
0 & 0 & 0 & 0 & 0 & 0 & 0 & 0 & 0 & 1 & 0 & 0& 0 & 0 & 0 & 0 & 0\\
0 & 0 & 0 & 0 & 0 & 0 & 0 & 0 & 0 & 0 & 1 & 0& 0 & 0 & 0 & 0 & 0\\
0 & 0 & 0 & 0 & 0 & 0 & 0 & 0 & 0 & 0 & 0 & 1& 0 & 0 & 0 & 0 & 0\\
0 & 0 & 0 & 0 & 0 & 0 & 0 & 0 & 0 & 0 & 0 & 0& 1 & 0 & 0 & 0 & 0\\
0 & 0 & 0 & 0 & 0 & 0 & 0 & 0 & 0 & 0 & 0 & 0& 0 & 1 & 0 & 0 & 0\\
0 & 0 & 0 & 0 & 0 & 0 & 0 & 0 & 0 & 0 & 0 & 0& 0 & 0 & 1 & 0 & 0\\
0 & 0 & 0 & 0 & 0 & 0 & 0 & 0 & 0 & 0 & 0 & 0& 0 & 0 & 0 & 1 & 0\\
0 & 0 & 0 & 0 & 0 & 0 & 0 & 0 & 0 & 0 & 0 & 0& 0 & 0 & 0 & 0 & 1\\
\hline
{\color{red}\textbf{1}} & {\color{red}\textbf{0}} & {\color{red}\textbf{0}} & {\color{red}\textbf{0}} & {\color{red}\textbf{0}} & {\color{red}\textbf{0}} & {\color{red}\textbf{0}} & {\color{red}\textbf{0}} & {\color{red}\textbf{0}} & {\color{red}\textbf{0}} & |\color{green}{\textbf{1}} &\color{green}{\textbf{0}}& \color{green}{\textbf{0}} & \color{green}{\textbf{0}} & \color{green}{\textbf{0}} & \color{green}{\textbf{0}} & \color{green}{\textbf{0}}\\
{\color{red}\textbf{0}} & {\color{red}\textbf{1}} & {\color{red}\textbf{0}} & {\color{red}\textbf{0}} & {\color{red}\textbf{0}} & {\color{red}\textbf{0}} & {\color{red}\textbf{0}} & {\color{red}\textbf{0}} & {\color{red}\textbf{0}} & {\color{red}\textbf{0}} & |\color{green}{\textbf{0}} & \color{green}{\textbf{1}}& \color{green}{\textbf{0}} & \color{green}{\textbf{0}} &\color{green}{\textbf{0}} & \color{green}{\textbf{0}} & \color{green}{\textbf{0}}\\
{\color{red}\textbf{0}} & {\color{red}\textbf{0}} & {\color{red}\textbf{1}} & {\color{red}\textbf{0}} & {\color{red}\textbf{0}} & {\color{red}\textbf{0}} & {\color{red}\textbf{0}} & {\color{red}\textbf{0}} & {\color{red}\textbf{0}} & {\color{red}\textbf{0}} & |\color{green}{\textbf{0}} & \color{green}{\textbf{0}}& \color{green}{\textbf{1}} & \color{green}{\textbf{0}} & \color{green}{\textbf{0}} & \color{green}{\textbf{0}} & \color{green}{\textbf{0}}\\
{\color{red}\textbf{0}} & {\color{red}\textbf{0}} & {\color{red}\textbf{0}} & {\color{red}\textbf{1}} & {\color{red}\textbf{0}} & {\color{red}\textbf{0}} & {\color{red}\textbf{0}} & {\color{red}\textbf{0}} & {\color{red}\textbf{0}} & {\color{red}\textbf{0}} & |\color{green}{\textbf{0}} & \color{green}{\textbf{0}}& \color{green}{\textbf{0}} & \color{green}{\textbf{1}} & \color{green}{\textbf{0}} & \color{green}{\textbf{0}} & \color{green}{\textbf{0}}\\
{\color{red}\textbf{0}} & {\color{red}\textbf{0}} & {\color{red}\textbf{0}} & {\color{red}\textbf{0}} & {\color{red}\textbf{1}} & {\color{red}\textbf{0}} & {\color{red}\textbf{0}} & {\color{red}\textbf{0}} & {\color{red}\textbf{0}} & {\color{red}\textbf{0}}& |\color{green}{\textbf{0}} & \color{green}{\textbf{0}}& \color{green}{\textbf{0}} & \color{green}{\textbf{0}} & \color{green}{\textbf{1}} & \color{green}{\textbf{0}} & \color{green}{\textbf{0}}\\
{\color{red}\textbf{0}} & {\color{red}\textbf{0}} & {\color{red}\textbf{0}} & {\color{red}\textbf{0}} & {\color{red}\textbf{0}} & {\color{red}\textbf{1}} & {\color{red}\textbf{0}} & {\color{red}\textbf{0}} & {\color{red}\textbf{0}} & {\color{red}\textbf{0}} & |\color{green}{\textbf{0}} & \color{green}{\textbf{0}}& \color{green}{\textbf{0}} & \color{green}{\textbf{0}} & \color{green}{\textbf{0}} & \color{green}{\textbf{1}} & \color{green}{\textbf{0}}\\
{\color{red}\textbf{0}} & {\color{red}\textbf{0}} & {\color{red}\textbf{0}} & {\color{red}\textbf{0}} & {\color{red}\textbf{0}} & {\color{red}\textbf{0}} & {\color{red}\textbf{1}} & {\color{red}\textbf{0}} & {\color{red}\textbf{0}} & {\color{red}\textbf{0}} & |\underline{\color{green}{\textbf{0}}} & \underline{\color{green}{\textbf{0}}}& \underline{\color{green}{\textbf{0}}} & \underline{\color{green}{\textbf{0}}} & \underline{\color{green}{\textbf{0}}} & \underline{\color{green}{\textbf{0}}} & \underline{\color{green}{\textbf{1}}}\\
{\color{red}\textbf{0}} & {\color{red}\textbf{0}} & {\color{red}\textbf{0}} & {\color{red}\textbf{0}} & {\color{red}\textbf{0}} & {\color{red}\textbf{0}} & {\color{red}\textbf{0}} & {\color{red}\textbf{1}} & {\color{red}\textbf{0}} & {\color{red}\textbf{0}} & |\color{blue}{\textbf{1}} & \color{blue}{\textbf{0}}& \color{blue}{\textbf{0}} & \color{blue}{\textbf{1}} & \color{blue}{\textbf{0}} & \color{blue}{\textbf{0}} & |\textbf{1}\\
{\color{red}\textbf{0}} & {\color{red}\textbf{0}} & {\color{red}\textbf{0}} & {\color{red}\textbf{0}} & {\color{red}\textbf{0}} & {\color{red}\textbf{0}} & {\color{red}\textbf{0}} & {\color{red}\textbf{0}} & {\color{red}\textbf{1}} & {\color{red}\textbf{0}} & |\color{blue}{\textbf{0}} & \color{blue}{\textbf{1}}& \color{blue}{\textbf{0}} & \color{blue}{\textbf{0}} & \color{blue}{\textbf{1}} & \color{blue}{\textbf{0}} & |\textbf{1}\\
{\color{red}\textbf{0}} & {\color{red}\textbf{0}} & {\color{red}\textbf{0}} & {\color{red}\textbf{0}} & {\color{red}\textbf{0}} & {\color{red}\textbf{0}} & {\color{red}\textbf{0}} & {\color{red}\textbf{0}} & {\color{red}\textbf{0}} & {\color{red}\textbf{1}} & |\color{blue}{\textbf{0}} & \color{blue}{\textbf{0}}& \color{blue}{\textbf{1}} & \color{blue}{\textbf{0}} & \color{blue}{\textbf{0}} & \color{blue}{\textbf{1}} & |\textbf{1}\\
\end{array}
\right]$$

}
\end{small}
\caption{Encoding matrix for the MUICP in Example \ref{ex4}.}
\label{ex4matrix}
\end{figure*}
\end{example}
\subsection{Main Results}
\begin{theorem}
\label{thm1}
Consider a symmetric MUICP with $K$ messages and $K$ receivers. Receiver $R_k$ wants the message $x_k$ and its antidotes are given by \eqref{antidote1}. For this index coding problem, the matrix $\mathbf{L}_{K \times (K-D)}$ given by {\bf CONSTRUCTION} as in Fig. \ref{fig1} and Fig. \ref{fig2} is an encoding matrix over every field $\mathbb{F}_q$ and the code generated is of optimal length.
\end{theorem}
\begin{proof}
If $\lambda_l$ divides $\lambda_{l-1}$ and $l$ is an even integer, then the sequence of construction of the matrices is given by \eqref{seq1}. The matrix $\mathbf{C}_{\lambda_{l-1} \times \lambda_{l}}^{\mathsf{T}}$ is a rectangular circulant matrix and every adjacent $\lambda_l$ columns of this matrix are linearly independent. In the matrix $\mathbf{D}_{\lambda_{l-2} \times \lambda_{l-1}}^{(2)}$, every $\lambda_{l-1}$ rows are linearly independent (Definition \ref{def1}). Similarly, in the matrix $\mathbf{D}_{\lambda_{l-2} \times \lambda_{l-3}}^{(3)}$, every $\lambda_{l-2}$ columns are linearly independent (Lemma \ref{lemma3}). The sequence of construction of the matrices in \eqref{seq1} guarantees that in the matrix $\mathbf{L}_{K \times (K-D)}$ given in Fig. \ref{fig1} every $K-D$ adjacent rows are linearly independent over every field $\mathbb{F}_q$. Thus, for the given MUICP, the matrix $\mathbf{L}_{K \times (K-D)}$ given in Fig. \ref{fig1} is an encoding matrix over every field $\mathbb{F}_q$ according to Lemma \ref{lemma1}. 

If $\lambda_l$ divides $\lambda_{l-1}$ and $l$ is an odd integer, then the sequence of construction of the matrices is given by \eqref{seq2}. The matrix $\mathbf{C}_{\lambda_{l-1} \times \lambda_{l}}$ is a rectangular circulant matrix and every $\lambda_l$ adjacent rows are linearly independent.~~In the matrix $\mathbf{D}_{\lambda_{l-1} \times \lambda_{l-2}}^{(2)}$,~~every $\lambda_{l-1}$ columns are linearly independent (Definition \ref{def1}). Similarly, in the matrix $\mathbf{D}_{\lambda_{l-3} \times \lambda_{l-2}}^{(3)}$, every $\lambda_{l-2}$ rows are linearly independent (Lemma \ref{lemma3}). The sequence of construction of the matrices in \eqref{seq2} guarantees that in the matrix $\mathbf{L}_{K \times (K-D)}$ given in Fig. \ref{fig2} every $K-D$ adjacent rows are linearly independent over every field $\mathbb{F}_q$. Thus, for the given MUICP, the matrix $\mathbf{L}_{K \times (K-D)}$ given in Fig. \ref{fig2} is an encoding matrix over every field $\mathbb{F}_q$ according to Lemma \ref{lemma1}. 

The encoding matrix $\mathbf{L}_{K \times (K-D)}$ encodes $K$ messages  into $K-D$ code symbols. Thus, the rate achieved by this code is $\frac{1}{K-D}$ and is equal to the capacity of the given MUICP. Hence, the code generated by the encoding matrix in Fig. \ref{fig1} and Fig. \ref{fig2} is of optimal length. 
\end{proof}
\begin{remark}
In \cite{BiK}, Birk and Kol defined partial clique and gave a coding scheme for a given index coding problem based on the partial cliques of the side information graph. A directed graph $G(V,E)$ is a $k$-partial clique $Clq(s,k)$ iff $\vert V \vert=s$, outdeg$(v) \geq (s-1-k)$, $\forall \ v \in V$, and there exists a $\ v \in V$ such that outdeg$(v)=(s-1-k)$. It can be observed that the side information graph of one-sided MUICP with $K$ messages and $D$ neighboring antidotes is a $(K-D-1)$-partial clique. An optimal index code for this $(K-D-1)$-partial clique can be obtained by using a $K-D$ erasure correcting MDS code. However, by using MDS codes, the size of the field depends on the number of messages $K$. The encoding matrix $\mathbf{L}_{K \times (K-D)}$ in {\bf CONSTRUCTION} is an encoding matrix over every field. Hence, the optimal length of a one-sided neighboring antidote MUICP is independent of field size. However, when an MDS code exists as an index coding problem with $K$ messages and $D$ antidotes  the antidotes need not be neighboring antidotes. 
\end{remark}
\begin{remark}
The capacity achieving code is a scalar linear code for a one-sided neighboring  antidote MUICP. The optimal length of a scalar linear code is called the \textit{minrank} of side information graph \cite{YBJK}. The optimal length of one-sided MUICP is independent of field size. Thus, the \textit{minrank} of the side information graph of a one-sided neighboring  antidote MUICP is independent of field size.
\end{remark}
\begin{remark}
The index code construction in Theorem \ref{thm1} enables each receiver $R_k$ to decode not only its required message $x_k$ but also all its interfering messages in $\mathcal{I}_k$ (Lemma \ref{lemma2}).
\end{remark}
\begin{lemma}
Let $D$ be a positive integer. Consider a symmetric MUICP with $K$ messages and $K$ receivers. Receiver $R_k$ wants the message $x_k$ and its antidotes are given by 
\begin{align*}
{\cal K}_k =\{x_{k+1}, x_{k+2},\dots,x_{k+D}\}. 
\end{align*}
The encoding matrix constructed in Theorem \ref{thm1} for this index coding problem can be written as 
\begin{align*} 
\mathbf{L}_{K \times (K-D)}=\begin{bmatrix}
\mathbf{I}_{K-D}\\
\mathbf{P}\\
\end{bmatrix},
\end{align*}
where $\mathbf{P}=\mathbf{D}_{D \times (K-D)}$. Consider a symmetric MUICP with $K$ messages and $K$ receivers. Receiver $R_k$ wants the message $x_k$ and its antidotes are given by 
\begin{align}
\label{dual}
{\cal K}_k =\{x_{k+1}, x_{k+2},\dots,x_{k+K-D}\}. 
\end{align}
Then, an encoding matrix for this index coding problem can be given as 
\begin{align*} 
\mathbf{L}_{K \times D}=\begin{bmatrix}
\mathbf{I}_{D}\\
\mathbf{P}^{\mathsf{T}}\\
\end{bmatrix}.
\end{align*}
\end{lemma}
\begin{proof}
In the matrix $\mathbf{P}$, every $D$ adjacent columns are linearly independent according to the construction procedure given in {\bf CONSTRUCTION}. Hence, in the matrix $\mathbf{P}^{\mathsf{T}}$, every $D$ adjacent rows are linearly independent. In the matrix $\mathbf{L}_{K \times D}$, every $D$ adjacent rows are linearly independent (Lemma \ref{lemma3}).   According to Lemma \ref{lemma1}, the matrix $\mathbf{L}_{K \times D}$ is an  encoding matrix for the MUICP with the antidotes given in \eqref{dual}.
\end{proof}

Consider the setting with $K$ receivers and $K$ messages and the $k$th receiver $R_{k}$ demanding the message $x_{k}$ having the side information given in \eqref{vantidotesa}. 
\begin{equation}
\label{vantidotesa}
{\cal K}_k= \{x_{k+t+1}, x_{k+t+2},\dots,x_{k+t+D}\}
\end{equation}
where $t \in \{0,1,2,\ldots,K-D-1\}$. If $t=0$, then this setting is equal to symmetric MUICP with neighboring side information. The cases of $t \neq 0$ corresponds to consecutive side-information but not neighboring.

\begin{lemma}
Consider a symmetric MUICP with $K$ messages and  $K$ receivers. Receiver $R_k$ wants the message $x_k$ and its antidotes are given by \eqref{vantidotesa}. For this index coding problem, the capacity is given by 
\begin{align*}
C=\frac{1}{K-D}~~\text{per~~message}.
\end{align*}
For this index coding problem, the matrix $\mathbf{L}_{K \times (K-D)}$ given by {\bf CONSTRUCTION} as in Fig. \ref{fig1} and Fig. \ref{fig2} is an encoding matrix over every field $\mathbb{F}_q$ and the generated code is of optimal length.
\end{lemma}
\begin{proof}
The sequence of construction of the matrices in {\bf CONSTRUCTION} guarantees that in the matrix $\mathbf{L}_{K \times (K-D)}$ given in Fig. \ref{fig1} and Fig. \ref{fig2} every $K-D$ adjacent rows are linearly independent over every field $\mathbb{F}_q$ (Lemma \ref{lemma3}). Thus, for the given MUICP, the matrix $\mathbf{L}_{K \times (K-D)}$ given in Fig. \ref{fig1} and Fig. \ref{fig2} is an encoding matrix over every field $\mathbb{F}_q$ according to Lemma \ref{lemma1}. The index code construction enables the receiver $R_k$ to decode all other $K-D$ messages which are not in its antidotes. The reciver $R_k$ can decode the message $x_{k+t}$. The message $x_{k+t}$ is the neighboring message to the $D$ adjacent antidotes given in \eqref{vantidotesa}. Hence, the capacity of symmetric MUICP with consecutive antidotes can not be greater than the capacity of symmetric MUICP with neighboring antidotes ($C \leq \frac{1}{K-D}$). The matrix $\mathbf{L}_{K \times (K-D)}$ given in Fig. \ref{fig1} and Fig. \ref{fig2} is an encoding matrix for the given symmetric MUICP with consecutive antidotes. Hence, the capacity of this index coding problem is atleast $\frac{1}{K-D}$ ($C \geq \frac{1}{K-D}$). This completes the proof of capacity. The encoding matrix $\mathbf{L}_{K \times (K-D)}$ encodes $K$ messages into $K-D$ code symbols. Thus, the rate achieved by this code is $\frac{1}{K-D}$ and is equal to the capacity of the given MUICP. Hence, the code generated by the encoding matrix in Fig. \ref{fig1} and Fig. \ref{fig2} is of optimal length.  
\end{proof}

Consider the setting with $K$ receivers and $K$ messages and the $k$th receiver $R_{k}$ demanding the message $x_{k}$ having the side information given in \eqref{vantidotea}. 
\begin{equation}
\label{vantidotea}
{\cal K}_k= \{x_{k+t_k+1}, x_{k+t_k+2},\dots,x_{k+t_k+D}\}
\end{equation}
where $t_k \in \{0,1,2,\ldots,K-D-1\}$ and $k=\{1,2,\ldots,K\}$. If $t_1=t_2=\ldots=t_K$, then this setting is equal to symmetric MUICP with consecutive one-sided side information. If $t_1=t_2=\ldots=t_K=0$, then this setting is equal to symmetric MUICP with neighboring side information.
\begin{lemma}
Consider a MUICP with $K$ messages and  $K$ receivers. Receiver $R_k$ wants the message $x_k$ and its antidotes are given by \eqref{vantidotea}. In this index coding problem, the $D$ antidotes are not symmetric but consecutive. The matrix $\mathbf{L}_{K \times (K-D)}$ given by {\bf CONSTRUCTION} is an encoding matrix for this index coding problem.
\end{lemma}
\begin{proof}
The sequence of construction of the matrices in {\bf CONSTRUCTION} guarantees that in the matrix $\mathbf{L}_{K \times (K-D)}$ given in Fig. \ref{fig1} and Fig. \ref{fig2} every $K-D$ adjacent rows are linearly independent over every field $\mathbb{F}_q$ (Lemma \ref{lemma3}). In this MUICP, every receiver has $D$ consecutive antidotes. Thus, for the given MUICP, the matrix $\mathbf{L}_{K \times (K-D)}$ given in Fig. \ref{fig1} and Fig. \ref{fig2} is an encoding matrix over every field $\mathbb{F}_q$ according to Lemma \ref{lemma1} and Lemma \ref{lemma2}. 
\end{proof}
\begin{remark}
The capacity of symmetric MUICP with neighboring antidotes is given in \eqref{capacity}. The $K$ messages in this index coding problem are taking values from a field uniformly and independent. Let $m$ be an integer relatively prime to $K$. Define the mapping $\pi$ as
\begin{align*}
\pi: k\rightarrow mk
\end{align*}
for $k=1,2,\ldots,K$.
This mapping creates a new index coding problem with $K$ messages and $K$ receivers. Receiver $R_{\pi(k)}$ wants the message $x_{\pi(k)}$ and its antidotes are given by 
\begin{align}
\label{antidoteper}
\mathcal{K}_k=\{x_{\pi(k+1)},x_{\pi(k+2)},\ldots,x_{\pi(k+D)}\}
\end{align}
This new index coding problem is not a consecutive antidote MUICP. We can generate $\phi(K)$(Euler's totient function) such new index coding problems. The side information graph of this new MUICP is isomorphic to the side information graph of symmetric one sided neighboring antidote MUICP. Hence, the capacity of the new MUICP is equal to that of the symmetric MUICP with neighboring antidotes. Let the rows of the matrix $\mathbf{L}_{K \times (K-D)}$ given by {\bf CONSTRUCTION} be $\{L_1,L_2,\ldots,L_K\}$. Define the matrix $\mathbf{L}_{K \times (K-D)}^{\pi}$ to be a $K \times (K-D)$ matrix with rows $\{L_{\pi(1)},L_{\pi(2)},\ldots,L_{\pi(K)}\}$. The matrix $\mathbf{L}_{K \times (K-D)}^{\pi}$ is an encoding matrix for the MUICP with antidotes given by \eqref{antidoteper}.
\end{remark}
\begin{example}
Consider a one-sided neighboring antidote MUICP with $K=7, D=3$. For this index coding problem, we have $\lambda_1=1$ and $l=1.$ The encoding matrix $\mathbf{L}_{7 \times 4}$ is given below.
\arraycolsep=3pt
\setlength\extrarowheight{-2.0pt}
{
$$\mathbf{L}_{7 \times 4}=\left[\begin{array}{*{20}c}
   1 & 0 & 0 & 0 \\
   0 & 1 & 0 & 0 \\
   0 & 0 & 1 & 0 \\
   0 & 0 & 0 & 1 \\
  \color{blue}{\textbf{1}} & \color{blue}{\textbf{0}} & \color{blue}{\textbf{0}} & \color{red}{\textbf{1}} \\
  \color{blue}{\textbf{0}} & \color{blue}{\textbf{1}} & \color{blue}{\textbf{0}} & \color{red}{\textbf{1}}\\
  \color{blue}{\textbf{0}} & \color{blue}{\textbf{0}} & \color{blue}{\textbf{1}} & \color{red}{\textbf{1}} \\
   \end{array}\right]$$
}
Let $m=2$. The mapping $\pi$ is given below.

\begin{table}[h]
\centering
\setlength\extrarowheight{2.2pt}
\begin{tabular}{|c|c|c|c|c|c|c|c|}
\hline
\textbf{$k$} & $1$ & $2$ & $3$ & $4$ & $5$ & $6$ & $7$\\
\hline
\textbf{$\pi(k)$} & $2$ & $4$ & $6$ & $1$ & $3$ & $5$ & $7$ \\
\hline  
\end{tabular}
\label{experm1}
\vspace{5pt}
\caption{}
\end{table}
\vspace{-10pt}
The new MUICP is described in the following table.
\begin{table}[h]
\centering
\setlength\extrarowheight{2.2pt}
\begin{tabular}{|c|c|c|}
\hline
\textbf{Rx} & {$\mathcal{W}_k$} & $\mathcal{K}_k$\\
\hline
\textbf{$R_1$} & $x_1$ & $x_3,x_{5},x_7$ \\
\hline
\textbf{$R_2$} & $x_2$& $x_{4},x_6,x_{1}$ \\
\hline 
\textbf{$R_3$} & $x_3$ & $x_{5},x_7,x_{2}$\\
\hline
\textbf{$R_4$} & $x_4$ & $x_6,x_{1},x_{3}$  \\
\hline
\textbf{$R_5$} & $x_5$ & $x_7,x_{2},x_4$  \\
\hline
\textbf{$R_6$} & $x_6$ & $x_{1},x_3,x_{5}$  \\
\hline
\textbf{$R_7$} & $x_7$ & $x_{2},x_4,x_{6}$  \\
\hline 
\end{tabular}
\vspace{5pt}
\caption{}
\label{experm2}
\end{table}
\vspace{-10pt}

The encoding matrix $\mathbf{L}_{7 \times 4}^{\pi}$ for the new MUICP defined in Table \ref{experm2} is given below.
\arraycolsep=3pt
\setlength\extrarowheight{-2.0pt}
{
$$\mathbf{L}_{7 \times 4}^{\pi}=\left[\begin{array}{*{20}c}
   0 & 0 & 0 & 1 \\
   1 & 0 & 0 & 0 \\
   \color{blue}{\textbf{1}} & \color{blue}{\textbf{0}} & \color{blue}{\textbf{0}} & \color{red}{\textbf{1}} \\
   0 & 1 & 0 & 0 \\
     \color{blue}{\textbf{0}} & \color{blue}{\textbf{1}} & \color{blue}{\textbf{0}} & \color{red}{\textbf{1}}\\
   0 & 0 & 1 & 0 \\
  \color{blue}{\textbf{0}} & \color{blue}{\textbf{0}} & \color{blue}{\textbf{1}} & \color{red}{\textbf{1}} \\
   \end{array}\right]$$
}
\end{example}
\subsection{Optimal length index codes for two-sided antidote MUICPs}
 In \cite{MRarXiv}, we proposed a vector linear index code construction which constructs a vector linear index code for two-sided  neighboring antidote MUICPs starting from a given one-sided neighboring antidote MUICP with a known scalar linear index code. The construction given in Theorem \ref{thm1} along with our construction in \cite{MRarXiv} gives a capacity achieving vector linear index codes for two-sided neighboring antidote problems for every $K,U$ and $D$.  If a scalar linear index code $\mathfrak{C}$ for the one-sided antidote problem is defined in the field $\mathbb{F}_q$, the construction procedure in \cite{MRarXiv} gives the construction of vector linear code for the two-sided antidote index coding problem in the same field $\mathbb{F}_q$. Thus, the construction given in Theorem \ref{thm1} along with our construction in \cite{MRarXiv} gives capacity achieving  vector linear codes for two-sided neighboring antidote problems for every $K,U$ and $D$. The index codes so constructed are independent of field size.

\section{conclusion}
In this paper, a capacity achieving scalar linear coding scheme is proposed for one-sided neighboring antidotes MUICPs. Some of the interesting directions of further research are as follows:
\begin{itemize}
\item Recently, it has been observed that in a noisy index coding problem it is desirable for the purpose of reducing the probability of error that the receivers use as small a number of transmissions from the source as possible and linear index codes with this property have been reported in \cite{TRCR}, \cite{KaR}. While the report \cite{TRCR} considers fading broadcast channels, in \cite{AnR1} and \cite{AnR2} AWGN channels are considered and it is reported that linear index codes with minimum length (capacity achieving codes or optimal length codes) help to facilitate to achieve more reduction in probability of error compared to non-minimum length codes for receivers with large amount of side-information. These aspects remain to be investigated for the constructed class of scalar linear codes.
\item In this paper, we proved that for the one-sided neighboring antidote symmetric MUICPs, the minrank of the side information graph is independent of the field size. We conjecture that for any unicast index coding problem with independent messages, the minrank of the side information graph is independent of the field size.
\end{itemize} 
\section*{Acknowledgment}
The authors would like to thank Lakshmi Prasad Natarajan for useful suggestions and pointing out a mistake in Lemma \ref{lemma3} in the previous version.

This work was supported partly by the Science and Engineering Research Board (SERB) of Department of Science and Technology (DST), Government of India, through J.C. Bose National Fellowship to B. Sundar Rajan.

\end{document}